\documentclass[conference]{IEEEtran}
\IEEEoverridecommandlockouts
\usepackage{bm}
\usepackage{cite}
\usepackage{amsmath,amssymb,amsfonts,graphicx,amsthm,mathtools,nicefrac}
\usepackage{algorithmic}
\usepackage{graphicx}
\usepackage{subfigure}
\usepackage{textcomp}
\usepackage{xcolor}
\def\BibTeX{{\rm B\kern-.05em{\sc i\kern-.025em b}\kern-.08em
    T\kern-.1667em\lower.7ex\hbox{E}\kern-.125emX}}

\allowdisplaybreaks
\newcommand\numberthis{\addtocounter{equation}{1}\tag{\theequation}}

\newtheorem{theorem}{Theorem}[section]
\newtheorem{assumption}{Assumption}[section]
\newtheorem{remark}{Remark}[section]

\numberwithin{equation}{section}

\newtheorem{lemma}[theorem]{Lemma}
\newtheorem{corollary}[theorem]{Corollary}

\DeclareMathOperator{\diag}{\mathrm{diag}}

\def\lsb{\left[}
\def\rsb{\right]}
\def\lab{\left|}
\def\rab{\right|} 
\def\la {\left\langle}
\def\ra {\right\rangle} 
\def \lb{\left(}
\def \rb{\right)}

\newcommand{\matsnorm}[2]{\left\| #1\right\|_{{#2}}}

\newcommand{\nucnorm}[1]{\ensuremath{\matsnorm{#1}{\footnotesize{\mbox{$\ast$}}}}}
\newcommand{\fronorm}[1]{\ensuremath{\matsnorm{#1}{\footnotesize{\mathsf{F}}}}}
\newcommand{\opnorm}[1]{\ensuremath{\matsnorm{#1}{}}}

\newcommand{\twonorm}[1]{\ensuremath{\matsnorm{#1}{\footnotesize{2}}}}

\newcommand{\bfm}[1]{\bm{#1}}

\def\va{\bfm a}   \def\mA{\bfm A}  
\def\vb{\bfm b}   \def\mB{\bfm B}  
     \def\C{\mathbb{C}}
\def\vd{\bfm d}   \def\mD{\bfm D}  
\def\ve{\bfm e}   \def\mE{\bfm E}  
    
\def\vg{\bfm g}     
\def\vh{\bfm h}     
   \def\mI{\bfm I}

   \def\mM{\bfm M}

     \def\R{\mathbb{R}}
\def\vs{\bfm s}   \def\mS{\bfm S}  
     
\def\vu{\bfm u}   \def\mU{\bfm U}  
\def\vv{\bfm v}   \def\mV{\bfm V}  
\def\vw{\bfm w}     
\def\vx{\bfm x}   \def\mX{\bfm X}  
\def\vy{\bfm y}   \def\mY{\bfm Y}  
   \def\mZ{\bfm Z}

\def\calA{{\cal  A}}

\def\calD{{\cal  D}}

\def\calG{{\cal  G}} 
\def\calH{{\cal  H}} 
\def\calI{{\cal  I}}

\def\calP{{\cal  P}}

\def\calX{{\cal  X}}

\newcommand{\bfsym}[1]{\bm{#1}}

             \def\bSigma{\bfsym \Sigma}
\def\blambda {\bfsym {\lambda}}

\def \calGT {\calG^\ast}

\def \tran {\mathsf{T}}
\def \tranH{\mathsf{H}}

\def \bzero{\bm 0}

\usepackage{bbm}

\begin{document}

\title{Simultaneous Blind Demixing and Super-resolution via Vectorized Hankel Lift\thanks{Corresponding authors: Jinchi Chen and Li Yu.}
\thanks{Jinchi Chen was partially supported by National Science Foundation of China under Grant No. 12001108.}
}

\author{
    \IEEEauthorblockN{Haifeng Wang\IEEEauthorrefmark{1},
    Jinchi Chen\IEEEauthorrefmark{2},
    Hulei Fan\IEEEauthorrefmark{1},
    Yuxiang Zhao\IEEEauthorrefmark{3},
    Li Yu\IEEEauthorrefmark{3}}
    \IEEEauthorblockA{\IEEEauthorrefmark{1}China Mobile (Zhejiang) Research \& Innovation Institute, Hangzhou, China
    \\\{wanghaifeng40, fanhulei\}@zj.chinamobile.com}
    \IEEEauthorblockA{\IEEEauthorrefmark{2}School of Mathematics, East China University of Science and Technology, Shanghai, China
    \\\{jcchen.phys\}@gmail.com}
    \IEEEauthorblockA{\IEEEauthorrefmark{3}China Mobile Research Institute, Beijing, China
    \\\{zhaoyuxiang, yuliyf\}@chinamobile.com}
}

\maketitle

\begin{abstract}
In this work, we investigate the problem of simultaneous blind demixing and super-resolution. Leveraging the subspace assumption regarding unknown point spread functions, this   problem can be reformulated as a low-rank matrix demixing problem. We propose a convex recovery approach that utilizes the low-rank structure of each vectorized Hankel matrix associated with the target matrix. Our analysis reveals that for achieving exact recovery, the number of samples needs to satisfy the condition $n\gtrsim Ksr \log (sn)$. Empirical evaluations demonstrate the recovery capabilities and the computational efficiency of the convex method.
\end{abstract}


\section{Introduction}
The simultaneous blind demixing and super-resolution of point sources refers to the problem of concurrently achieving blind super-resolution \cite{chi2016guaranteed,yang2016super,li2019atomic,chen2022vectorized} for $K$ point source signals within their superimposed mixture. This problem arises in a range of applications, including but not limited to joint radar-communications \cite{vargas2023dual}, multi-user multi-channel estimation \cite{luo2006low}.

It is well-established that blind super-resolution is intrinsically ill-posed without additional assumptions \cite{chi2016guaranteed,yang2016super,li2019atomic,chen2022vectorized}. The problem under consideration can be viewed as an extension of blind super-resolution, which exacerbates its complexity. Consequently, we  introduce a subspace assumption and reformulate the problem of simultaneous blind demixing and super-resolution as a structured low-rank matrix demixing problem.


Recent research efforts, as demonstrated in \cite{vargas2023dual, monsalve2022beurling, jacome2023multi, razavikia2023off}, have harnessed the inherent structure of data matrices to develop various convex relaxation techniques for addressing this problem. Specifically, in \cite{vargas2023dual}, inspired by joint radar and communication systems, 
the authors proposed an atomic norm minimization (ANM) approach for simultaneous blind demixing and super-resolution with $K=2$. Furthermore, in \cite{monsalve2022beurling}, a nuclear norm minimization method was designed for the same problem. Additionally, \cite{razavikia2023off} extended ANM to address this problem for arbitrary $K$, but the theoretical analysis is  lacking.

In this paper, our focus centers on addressing the simultaneous blind demixing and super-resolution problem for arbitrary values of $K$. We utilize the vectorized Hankel lift technique as introduced in \cite{chen2022vectorized} to leverage the low-dimensional structures within the target matrices. This approach enables a convex framework for reconstruction, and the exact recovery guarantees based on standard assumptions are established.

\section{Problem formulation and proposed method}
\subsection{Problem Formulation}

Consider a set of $K$ point source signals denoted as $\{x_k(t)\}_{k = 1}^K$, where the $k$-th signal can be expressed in the following form:
\begin{align*}
	x_k(t) = \sum_{\ell = 1}^{r_k} d_{k,\ell} \delta(t-\tau_{k,\ell}),
\end{align*}
where $\delta(\cdot)$ represents the Dirac function, $r_k$ indicates the number of spikes, and $\{\tau_{k,\ell}\}_{\ell = 1}^{r_k}$ and $\{d_{k,\ell}\}_{\ell =1}^{r_k}$ represent the locations and amplitudes of the point source signals, respectively.

Let $y(t)$ be the summation of point source signals convolved with unknown point spread functions, given by
\begin{align}
\label{eq: receive signal}
	y(t) = \sum_{k = 1}^K x_k(t)\ast g_k(t) = \sum_{k=1}^{K}   \sum_{\ell = 1}^{r_k} d_{k,\ell} g_{k}(t-\tau_{k,\ell}).
\end{align}
By taking the Fourier transform of~\eqref{eq: receive signal} and subsequently sampling, we obtain that for $j=0,\cdots, n-1,$
\begin{align}\label{eq: observation}
	 \vy[j] = \sum_{k=1}^{K} \sum_{\ell = 1}^{r_k} d_{k,\ell} e^{-\imath 2\pi (j-1)\tau_{k,\ell}} \hat{g}_{k}[j],
\end{align}
where $\vg_{k} := \begin{bmatrix}
	\hat{g}_k[0] &\cdots & \hat{g}_k[n-1] 
\end{bmatrix}^\tran\in\C^n$
are unknown. The goal of the simultaneous blind demixing and super-resolution problem is to jointly recover both $\{d_{k,\ell}, \tau_{k,\ell}\}$ and $\{\vg_k\}$ from~\eqref{eq: observation}.

As previously mentioned, simultaneous blind demixing and super-resolution is an ill-posed problem without any additional assumptions. In alignment with prior research \cite{vargas2023dual,monsalve2022beurling,jacome2023multi,razavikia2023off}, we adopt a similar approach and make the assumption that each point spread function $\vg_k$ lies within a known subspace defined by $\mB_k\in\C^{n\times s_k}$, such that:
\begin{align*}
    \vg_k = \mB_k\vh_k, \quad k=1,\cdots, K,
\end{align*}
where $\vh_k\in\C^{s_k}$ denotes an unknown coefficient vector. 

Under the subspace assumption and by leveraging a lifting approach, we can express the measurements given in~\eqref{eq: observation} as a superposition of linear observations involving $\left\{\mX_k^\natural := \sum_{\ell = 1}^{r_k} d_{k,\ell} \vh_k\va_{\tau_{k,\ell}}^\tran \right\}_{k=1}^K$:
\begin{align*}
    \vy[j] = \sum_{k=1}^{K} \la \ve_j\vb^\tranH_{k,j}, \mX_k^\natural \ra,\quad j=0,\cdots, n-1,
\end{align*}
where $\ve_j\in\R^n$ represents the $j$-th standard basis vector in $\R^n$, $\vb_{k,j}$ denotes the $j$-th row of $\mB_k$, and $\va_{\tau}\in\C^n$ stands as the steering vector, defined as
\begin{align*}
\begin{bmatrix}
1 & e^{-\imath2\pi \cdot 1\cdot \tau} &\cdots & e^{-\imath2\pi \cdot (n-1)\cdot \tau} 
\end{bmatrix}^\tran\in\C^{n}.
\end{align*}
Without loss of generality, we assume that $s_1 = \cdots = s_K = s$ and $r_1 = \cdots= r_K=r$. Additionally, consider $\calA_k$ as a linear operator $\calA_k:\C^{s\times n}\rightarrow \C^n$, defined as
\begin{align*}
	\left(\calA_k\lb \mX_k^\natural\rb \right)[j] = \la  \ve_j\vb^\tranH_{k,j}, \mX_k^\natural\ra, \quad j=0,\cdots, n-1.
\end{align*} 
Consequently, the measurement model can be succinctly expressed as
\begin{align}\label{eq: obs}
	\vy = \sum_{k=1}^{K} \calA_k\lb\mX_k^\natural\rb.
\end{align}
Hence, the problem of simultaneous blind demixing and super-resolution can be cast as the task of demixing a sequence of matrices $\{\mX_k^\natural\}$ from the superimposed linear measurements of these matrices. 
Upon successfully recovering the data matrices ${\mX_k^\natural}$, it becomes feasible to extract the frequencies ${\tau_{k,\ell}}$ through spatial smoothing MUSIC \cite{chen2022vectorized,evans1981high, evans1982application,yang2019source}.

\subsection{Proposed Method}
Let $\calH$ denote the vectorized Hankel lifting operator, which transforms a matrix $\mX\in\C^{s\times n}$ into a matrix of dimensions $sn_1\times n_2$, defined as follows:
\begin{align*}
    \calH(\mX) = \begin{bmatrix}
        \vx_0 &\vx_1&\cdots&\vx_{n_2-1}\\
        \vx_1 &\vx_2&\cdots&\vx_{n_2}\\
        \vdots &\vdots &\ddots &\vdots\\
        \vx_{n_1-1}&\vx_{n-1}&\cdots&\vx_{n-1}
    \end{bmatrix}\in\C^{sn_1\times n_2},
\end{align*}
where $\vx_i\in\C^s$ represents the $i$-th column of $\mX$, and $n_1+n_2 = n+1$. It has been demonstrated that the rank of $\calH(\mX_k^\natural)$ is at most $r$ \cite{chen2022vectorized}. Consequently, we adopt the nuclear norm minimization to promote the low rank structure of $\calH(\mX_k^\natural)$ and consider the following convex approach to recover $\{\mX_k^\natural\}$:
\begin{align}\label{eq: oripro}
    \min_{\{\mX_k\}_{k = 1}^K} \sum_{k = 1}^K \nucnorm{\calH(\mX_k)},\quad \mbox{s.t.} \sum_{k=1}^{K} \calA_k(\mX_k) = \vy,
\end{align}
which is denoted as the Multiple Vectorized Hankel Lift (MVHL). Since this convex optimization problem can be effectively tackled using various existing software packages, our focus is narrowed down to evaluating the theoretical performance of \eqref{eq: oripro} and investigating when its solution aligns with $\{\mX_k^\natural\}$.

\section{Main results}
Before presenting our main results, we will introduce some standard assumptions.
\begin{assumption}
\label{assumption 1}
    The matrices $\{\mB_k\}$ are independent, and the column vectors $\{\vb_{k,\ell}\}_{\ell = 1}^n$ of the subspace matrix $\mB_k$ are sampled independently and identically from a population $F$ which obeys the following properties:
    \begin{align}
        \label{isotropy}
        \mathbb{E}\lsb\vb\vb^\ast\rsb &= \mI_s \mbox{ and }
        \max_{0\leq k\leq s-1}|\vb[j]|^2\leq \mu_0.
    \end{align}
\end{assumption}
\begin{remark}
Assumption \ref{assumption 1} is widely used in blind super-resolution \cite{chi2016guaranteed,yang2016super,li2019atomic,chen2022vectorized}, dual-blind deconvolution \cite{vargas2023dual,monsalve2022beurling,jacome2023multi,razavikia2023off}, and can be satisfied by Rademacher random vector or when $\vb$ is uniformly sampled from a Discrete Fourier Transform matrix. 
\end{remark}

\begin{assumption}
\label{assumption 2}
	Let $\calH(\mX^\natural) = \mU\bSigma\mV^\tranH$ be the singular value decomposition of $\calH(\mX^\natural) $, where $\mU\in\C^{sn_1\times r}, \bSigma\in\R^{r\times r}$ and $\mV\in\C^{n_2\times r}$. Denote $\mU^\tranH = \begin{bmatrix}
		\mU_0^\tranH &\cdots &\mU_{n_1-1}^\tranH
	\end{bmatrix}^\tranH$, where $\mU_\ell= \mU[\ell s + 1 : (\ell+1) s, :]$ is the $\ell$-th block of $\mU$ for $\ell=0,\cdots, n_1-1$. The matrix $\mX^\natural$ is $\mu_1$-incoherence if $\mU$ and $\mV$  obey that
	\begin{align*}
		\max_{0\leq \ell\leq n_1-1} \fronorm{\mU_\ell}^2 \leq \frac{\mu_1 r}{n} \mbox{ and }\max_{0\leq j \leq n_2-1} \twonorm{\ve_j^\tran\mV}^2 \leq \frac{\mu_1 r}{n}
	\end{align*}
	for some positive constant $\mu_1 $.
\end{assumption}
\begin{remark}
Assumption \ref{assumption 2} is commonly adopted in blind super-resolution, and is satisfied when the minimum wrap-up distance between the locations of point sources is greater than about $2/n$.
\end{remark}

Now we state our main results, whose proofs are deferred to Section~\ref{sec: proof architecture}.
\begin{theorem}\label{thm: main theorem}
    Under Assumptions~\ref{assumption 1} and \ref{assumption 2}, the matrices $\{ \mX_k^\natural \}_{k = 1}^K$ are the  unique optimal solution to the problem \eqref{eq: oripro} with probability at least $1-c_0(sn)^{-c_1}$, provided that $n\gtrsim K\mu_0\mu_1sr\log(sn)$, where $c_0,c_1$ are absolute constants.
\end{theorem}

\section{Simulation Results}
In this section, we evaluate the empirical performance of MVHL for simultaneous blind demixing and super-resolution problem. To solve MVHL, we use the CVX optimization framework \cite{grant2014cvx}.

\subsection{Recovery Ability of MVHL}
We begin by investigate the recovery performance of MVHL in comparison to Atomic Norm Minimization (ANM) \cite{razavikia2023off} using the empirical phase transition framework. To generate the data matrices $\{\mX_k^\natural\}_{k = 1}^K$, we follow this procedure: the locations $\{\tau_{k,\ell}\}_{\ell}^L$ are uniformly sampled from the interval $[0,1)$, the amplitudes $\{d_{k,\ell}\}_{\ell}^L$ are set as $(1+10^{c_\ell})e^{-\imath\Psi_\ell}$, with $c_{\ell}$ uniformly sampled from $[0,1)$ and $\Psi_\ell$ uniformly sampled from $[0,2\pi)$; the coefficient $\vh_k$ is drawn from a standard Gaussian distribution with normalization. The subspace matrices $\{\mB_k\}_{k = 1}^K$ are independently sampled from the Discrete Fourier Transform (DFT) matrix. The experiments are conducted for $K = 2$, $n = 48$, and various values of $r$ and $s$. We execute each algorithm $20$ times for every combination of $r$ and $s$. A successful reconstruction is defined when the relative error satisfies $\sqrt{\frac{\sum_{k = 1}^K\fronorm{\mX_k^\natural-\mX_k}^2}{\sum_{k = 1}^K\fronorm{\mX_k^\natural}^2}}\leq 10^{-3}$.

Figure~\ref{fig: phase transition} (a) and Figure~\ref{fig: phase transition} (b) display the phase transition plots for MVHL and ANM, respectively. The figures reveal that MVHL exhibits a higher phase transition curve compared to ANM.


\begin{figure}[ht!]
	\centering
	\subfigure[]{
		\includegraphics[width=0.2\textwidth]{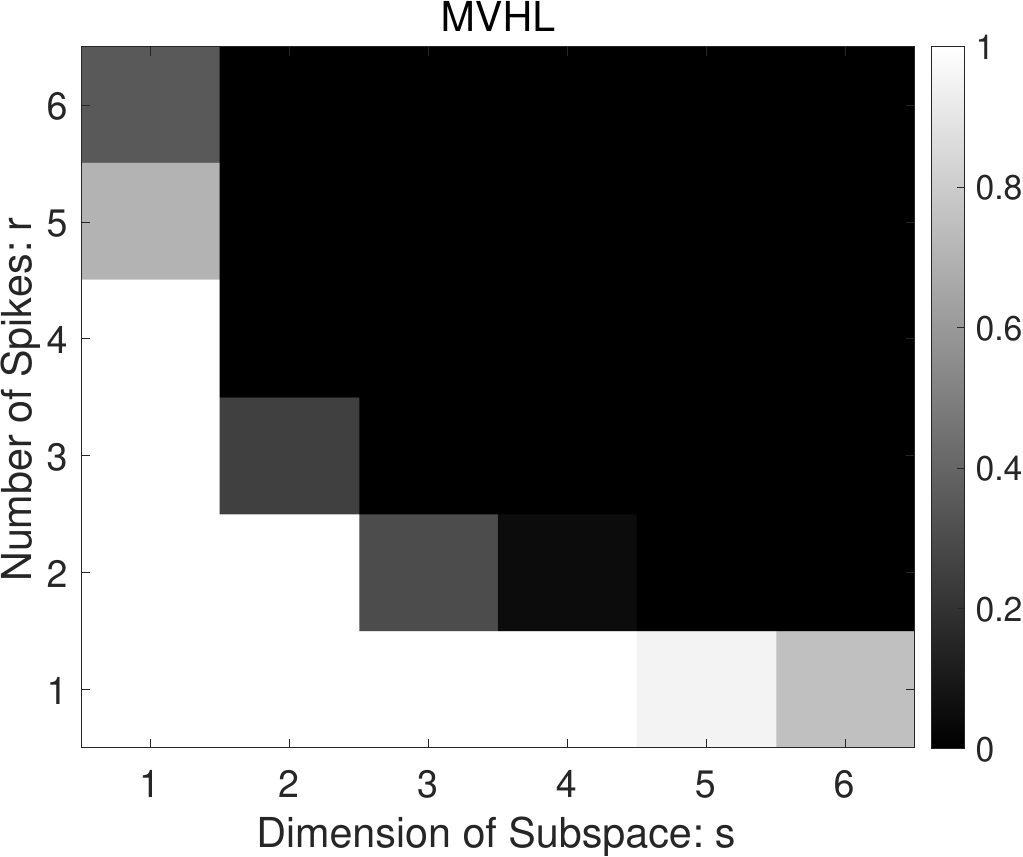}
	}
	\subfigure[]{
		\includegraphics[width=0.2\textwidth]{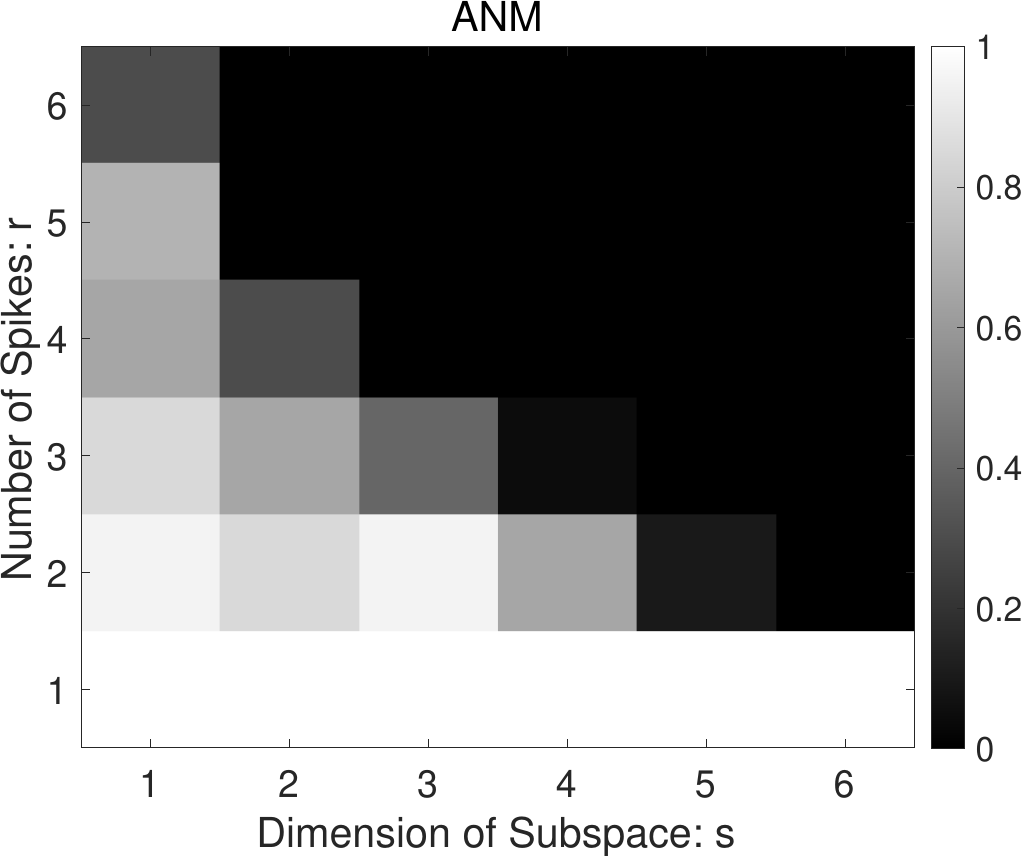}
	}
		
%
	
	\caption{The phase transitions of MVHL and  ANM}
	\label{fig: phase transition}
\end{figure}

\subsection{Robustness of MVHL}
In this experiment, we aim to illustrate the robustness of MVHL in the presence of additive noise.  Consider the additive noise model where a true signal $\vy$ is contaminated by a noisy vector $\ve$, given by:
\begin{align*}
    \ve = \varepsilon \cdot \twonorm{\vy} \cdot  \vw/ \twonorm{\vw},
\end{align*}
where $\sigma$ represents the noise level, and $\vw$ follows a standard multivariate normal distribution. We perform tests with values of $n$ equal to either $48$ or $64$, while setting $s$ and $r$ to 2. The noise level $\varepsilon$ is varied from $10^{-3}$ to $1$, corresponding to a signal-to-noise ratio (SNR) ranging from 60 dB to 0 dB. For each combination of $(n,\varepsilon)$, we conduct 10 random instances of the problem. Figure~\ref{fig: noise} displays the average relative error as a function of SNR. Notably, the plot clearly demonstrates a linear relationship between the relative reconstruction error and the noise level. Furthermore, it is evident that the relative reconstruction error decreases as the number of measurements increases.
\begin{figure}[ht!]
\centering
\includegraphics[width=0.3\textwidth]{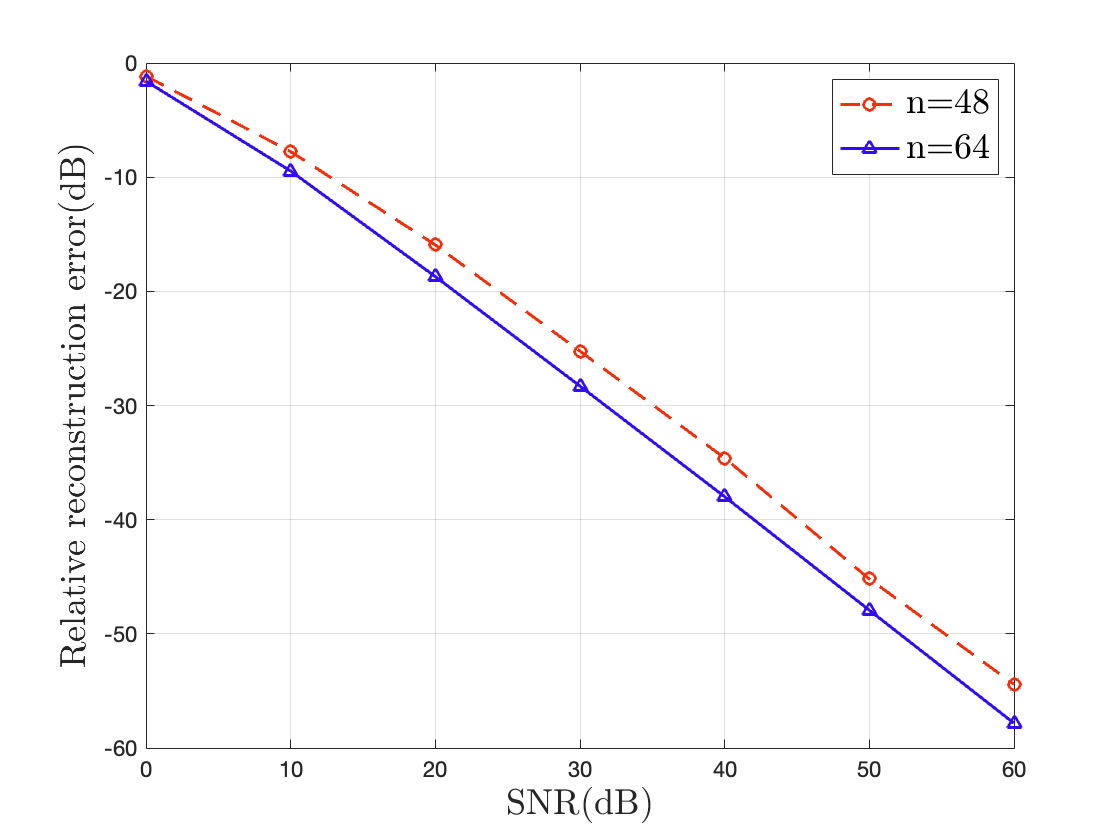}
\caption{Performance of MVHL under different noise levels}
\label{fig: noise}
\end{figure}

\subsection{Channel Parameter Estimation}
In this section, we evaluate the performance of MVHL in the context of joint radar-communication systems. The received signal is formulated as described in \cite{vargas2023dual}:
\begin{align*}
    \vy_p[n] &= \sum_{\ell = 1}^L [d_r]_{\ell} e^{-2\imath\pi(n[\tau_r]_{\ell}+p[\nu_r]_{\ell}) }\vs_{p}[n]\\
    &\quad +\sum_{j = 1}^J[d_c]_{j} e^{-2\imath\pi(n[\tau_c]_{j}+p[\nu_c]_{j}) }\vg_{p}[n],
\end{align*}
where $L$ and $J$ represent the number of targets and propagation paths, respectively. ${d_r}{\ell = 1}^L$ and ${d_c}{j = 1}^J$ denote the channel coefficients, ${\tau_r}{\ell = 1}^L$ and ${\tau_c}{j = 1}^J$ are the time delays, and ${\nu_r}{\ell = 1}^L$ and ${\nu_c}{j = 1}^J$ are the Doppler frequencies. $\vs_p[n]$ represents the discrete values of the Fourier transform of the $p$-th pulse $s(t-pT)$, while $\vg_{p}[n]$ denotes data symbols in the $p$-th message. Under the low-dimensional subspace assumption, and using arguments similar to those in \cite{mao2022blind}, the observed signal can be reformulated as:
\begin{align}\label{eq: ofdm}
    \vy[i] &=  \la \ve_i\vb_i^\tranH,\sum_{\ell = 1}^L [d_r]_\ell \vu[\va_{r}]_\ell^\tran \ra\notag\\
    &\quad + \la \ve_i\vd_i^\tranH,\sum_{j = 1}^J [d_r]_j \vv[\va_c]_j^\tran \ra,
\end{align}
for $j = 0,\cdots,NP-1$. Here  $\va_r$ and $\va_c$ are defined as
\begin{small}
\begin{align*}
    [\va_r]_\ell &= [e^{-2\imath\pi((0)[\tau_r]_{\ell}+(0)[\nu_r]_{\ell})},\cdots,e^{-2\imath\pi((N-1)[\tau_r]_{\ell}+(P-1)[\nu_r]_{\ell})}],\\
    [\va_c]_j &= [e^{-2\imath\pi((0)[\tau_c]_{j}+(0)[\nu_c]_{j})},\cdots,e^{-2\imath\pi((N-1)[\tau_c]_{j}+(P-1)[\nu_c]_{j})}].
\end{align*}
\end{small}
It's worth noting that~\eqref{eq: ofdm} is a special case of \eqref{eq: obs}.
In the simulation, we set $N = P = 10$ and $L = J = 2$. Each row of the subspace matrices is generated using the form $\begin{bmatrix}
       1 &e^{2\pi\imath f}&\cdots&e^{2\pi\imath(s-1) f}
   \end{bmatrix}$ with $f$ being uniformly sampled from $[0,1]$~\cite{chi2016guaranteed}. The target delays and Dopplers are drawn from the interval $[0,1]$ uniformly at random. The target matrices are recovered by solving problem~\eqref{eq: oripro}. Subsequently, the delays and Dopplers are estimated using 2D MUSIC.
   Figure~\ref{fig: channel} illustrates the true and estimated channel parameters for radar and communications. The results demonstrate that the proposed method, in conjunction with MUSIC, effectively recovers the delays and Dopplers.
\begin{figure}[ht!]
\centering
\includegraphics[width=0.3\textwidth]{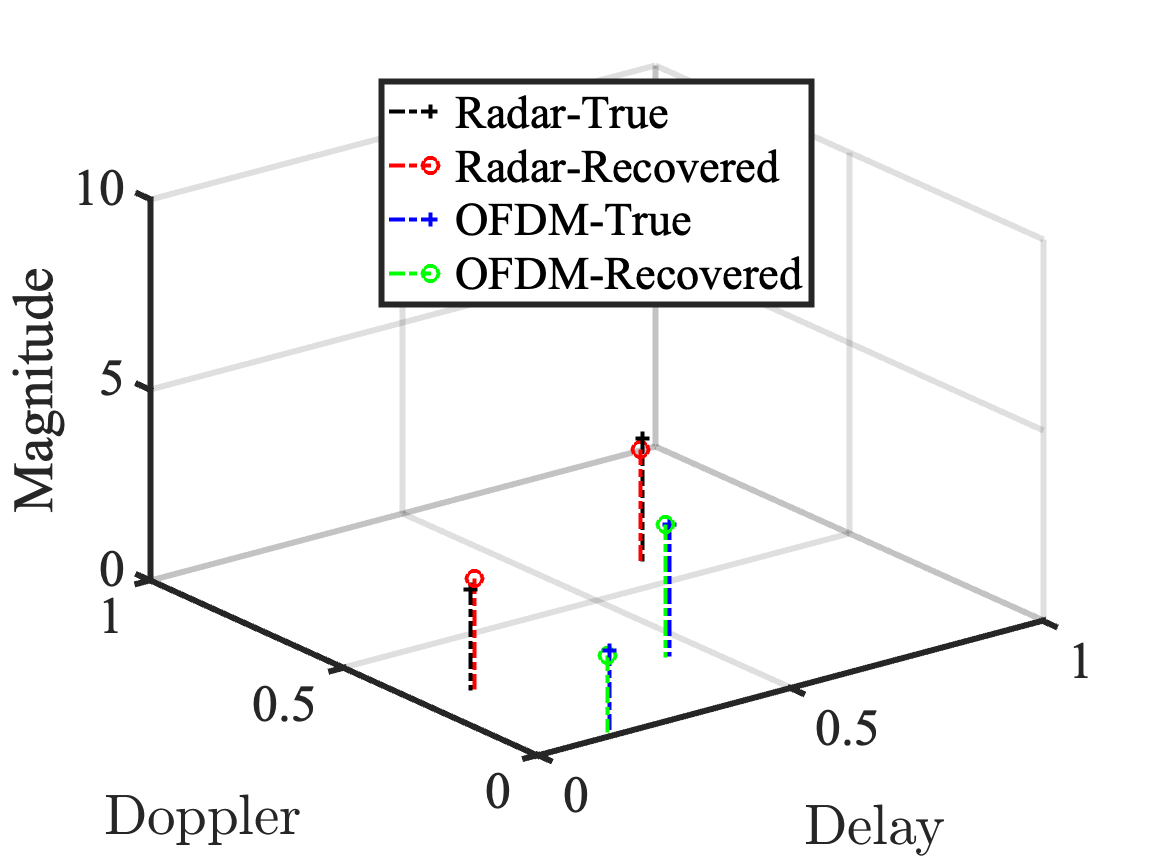}
\caption{Performance of MVHL for channel parameter estimation problem}
\label{fig: channel}
\end{figure}

\section{Proofs of Main Result}\label{sec: proof architecture}
The proof of Theorem \ref{thm: main theorem} relies on the dual certificate technique which has been widely used in analyzing low rank matrix recovery \cite{candes2009exact, gross2011recovering} and blind super-resolution \cite{chen2022vectorized}.  Let $\calH^\ast$ be the adjoint of $\calH$ which maps
$sn_1 \times n_2$ matrices to matrices of size $s \times n$. 
We define $\calD^2 = \calH^\ast\calH$. Given a matrix $\mX\in\C^{s\times n}$, the matrix $\calD^2\lb\mX\rb\in\C^{s\times n}$ is given by
\begin{align*}
    \calD^2\lb\mX\rb = \begin{bmatrix}
        w_0\vx_0 &\cdots&w_{n-1}\vx_{n-1}
    \end{bmatrix},
\end{align*}
where $w_i (i = 0,\cdots,n-1)$ is defined as 
\begin{align*}
    w_i = \sharp\{(j,k)~|~j+k = i,0\leq j\leq n_1-1,0\leq k\leq n_2-1\}.
\end{align*}
In addition, let the operator $\calG$ be $\calG = \calH\calD^{-1}$. The adjoint of $\calG$, denoted by $\calG^\ast$, is given by $\calG^\ast = \calD^{-1}\calH^\ast$.  
Letting $\mZ = \calH(\calX) = \calG\calD(\mX)$, one has $\calD(\mX) = \calG^\ast(\mZ)~\mbox{and}~ (\calI-\calG\calG^\ast)(\mZ) = \bzero$.
Moreover, define $\mD = \diag\lb\sqrt{w_0},\cdots,\sqrt{w_{n-1}}\rb$. We have $\calA_k\calD(\mX) = \mD\calA_k(\mX)$ for any matrix $\mX\in\C^{s\times n}$. Consequently, the problem~\eqref{eq: oripro} can be rewritten as 
\begin{align}\label{eq: voptp}
\min_{\{\mZ_k\}_{k = 1}^K} \sum_{k = 1}^K \nucnorm{\mZ_k} &\text{ s.t. }\sum_{k=1}^{K} \calA_k\calG^\ast(\mZ_k) = \mD\vy, \\
    &~ (\calI- \calG\calG^\ast)\lb\mZ_k\rb = \bzero, ~\mbox{for~}k = 1,\cdots,K.\notag
\end{align}
Due to the equivalence between \eqref{eq: oripro} and \eqref{eq: voptp}, it suffices to investigate the recovery guarantee of \eqref{eq: voptp}.

\subsection{Deterministic Optimality Condition}
Recall that the singular value decomposition (SVD) of $\calH(\mX_k^{\natural})$ is $\calH(\mX_k^{\natural}) = \mU_k\bSigma_k\mV_k^\tranH$. The tangent space $T_k$ is defined as
\begin{align*}
	T_k = \left\lbrace \mU_k\mA^\tranH+\mB\mV_k^\tranH:\mA\in\C^{n_2\times r},\mB\in\C^{sn_1\times r}\right\rbrace.
\end{align*}
The projection of  $\mM$ onto the tangent space $T_k$ is given by
\begin{align*}
    \calP_{T_k}(\mM) = \mU_k\mU_k^\tranH\mM+\mM\mV_k\mV_k^\tranH-\mU_k\mU_k^\tranH\mM\mV_k\mV_k^\tranH.
\end{align*}

\begin{theorem}\label{thm: unique theorem}
	Suppose for any $k\in[K]$,
	\begin{align}
 \label{eq: concentrate}
		\opnorm{\calP_{T_k}\calG\calA_k^\ast\calA_k\calG^\ast\calP_{T_k} -  \calP_{T_k}\calG\calG^\ast\calP_{T_k} } &\leq \frac{1}{4},\\
  \label{eq: incoh}
		\mu:=\max_{i\neq j}\opnorm{\calP_{T_i}\calG\calA_i^\ast\calA_j\calG^\ast\calP_{T_j}}&\leq \frac{1}{8K}.
	\end{align}
If there exists a series of $\{\mY_k\}_{k=1}^K\subset \C^{sn_1\times n_2}$ such that 
\begin{align}
\label{condition F norm}
\fronorm{\mU_k\mV_k^\ast-\calP_{T_k}( \mY_k)} &\leq \frac{1 }{48Ks\mu_0},\\
\label{condition O norm}
\opnorm{\calP_{T_k}^\perp (\mY_k)}&\leq \frac{1}{2},\\
\label{condition range}
\left(\calGT(\mY_1), \cdots, \calGT(\mY_K)\right)&\in\text{Range}(\calA^\ast).
\end{align}
Then $\left\{\mZ_k^\natural\right\}_{k=1}^K$ is the unique solution to problem~\eqref{eq: voptp}.
\end{theorem}
\begin{lemma}
	\label{lemma 1}
        Suppose the conditions~\eqref{eq: concentrate},~\eqref{eq: incoh}, and $\opnorm{\calA_k} \leq \sqrt{s\mu_0}$ holds.
	Assume $\mM_k$ obeys that
	\begin{align}\label{eq: feasible}
		\sum_{k=1}^K \calA_k \calG^\ast(\mM_k) = \bzero, \quad (\calI - \calG\calG^\ast)(\mM_k) = \bzero.
	\end{align}
	Then one has
	\begin{align*}
		\sum_{k=1}^K \fronorm{\calP_{T_k}(\mM_k)}^2  \leq 16Ks\mu_0 \sum_{k=1}^K  \fronorm{\calP_{T_k^\perp}(\mM_k)}^2.
	\end{align*}
\end{lemma}
\begin{proof}[Proof of Theorem~\ref{thm: unique theorem}] 
For any feasible solution of problem~\eqref{eq: voptp}, it must have the form of $\left\lbrace \mZ_k^{\natural}+\mM_k\right\rbrace_{k = 1}^K$, 
    where the perturbation matrices $\mM_k\in\C^{sn_1\times n_2}$ satisfy condition~\eqref{eq: feasible}.
 It suffices to show that 
	\begin{align*}
		\sum_{k = 1}^K \nucnorm{\mZ_k^{\natural}+\mM_k} > \sum_{k = 1}^K \nucnorm{\mZ_k^{\natural}}
	\end{align*}
	for any nontrivial set of $\left\lbrace \mM_k\right\rbrace_{k = 1}^K$. For each $\mM_k$, there  exists an   matrix $\mS_k\in T_k^{\perp}$ such that
	\begin{align*}
		\la \mM_k,\mS_k\ra = \nucnorm{\calP_{T_k^\perp}\lb\mM_k\rb} ~\mbox{and}~\opnorm{\mS_k}\leq 1.
	\end{align*}
	Therefore, $\mU_k\mV_k^\tran+\mS_k$  belongs to the sub-differential of $\nucnorm{\cdot}$ at $\mZ^\natural_k = \mU_k\bSigma_k\mV_k^\ast$.  By the definition of subgradient, one has
	\begin{align*}
		&\sum_{k = 1}^K \nucnorm{\mZ_k^{\natural}+\mM_k}\\
  &\geq \sum_{k = 1}^K\lb \nucnorm{\mZ_k^{\natural}}+\la \mU_k\mV_k^\ast+\mS_k,\mM_k\ra \rb\\
		& = \sum_{k = 1}^K\lb \nucnorm{\mZ_k^{\natural}}+\la \mU_k\mV_k^\ast,\mM_k\ra+ \nucnorm{\calP_{T_k^\perp}\lb\mM_k\rb} \rb\\
		&=\sum_{k = 1}^K\lb \nucnorm{\mZ_k^{\natural}}+\la \mU_k\mV_k^\ast-\mY_k, \mM_k\ra + \nucnorm{\calP_{T_k^\perp}\lb\mM_k\rb} \rb,
	\end{align*}
where the last is due to the fact that $\sum_{k=1}^{K}\la \mY_k, \mM_k \ra  = 0$. 
	To this end, we need to show that
	\begin{align*}
		I_1 := \sum_{k = 1}^K \left( \la \mU_k\mV_k^\ast-\mY_k,\mM_k\ra +   \nucnorm{\calP_{T_k^\perp}\lb\mM_k\rb} \right) >0.
	\end{align*}
	By decomposing the inner product in $I_1$ onto tangent space and its cotangent space respectively, one can obtain
	\begin{align*}
		I_1 &= \sum_{k = 1}^K \la \mU_k\mV_k^\ast-\calP_{T_k}\lb \mY_k \rb,\calP_{T_k}\lb\mM_k\rb\ra \\
  &\quad - \sum_{k= 1}^K\la\calP_{T_k}^\perp\lb \mY_k\rb,\calP_{T_k}^\perp\lb\mM_k\rb\ra+  \sum_{k = 1}^K\nucnorm{\calP_{T_k^\perp}\lb\mM_k\rb}\\
		&\geq -\sum_{k = 1}^K\fronorm{\calP_{T_k}\lb\mM_k\rb }\fronorm{\mU_k\mV_k^\ast-\calP_{T_k}\lb \mY_k\rb}\\
  &\quad +\sum_{k = 1}^K\nucnorm{\calP_{T_k^\perp}\lb\mM_k\rb}\lb 1-\opnorm{\calP_{T_k}^\perp\lb \mY_k\rb}\rb\\
		&\stackrel{(a)}{\geq} \lb-\frac{1}{48Ks\mu_0 } \cdot 16Ks\mu_0+\frac{1}{2} \rb\sum_{k=1}^K  \fronorm{\calP_{T_k^\perp}(\mM_k)}^2  \\
		&=\frac{1}{6}\sum_{k=1}^K  \fronorm{\calP_{T_k^\perp}(\mM_k)}^2 >0,
	\end{align*}
where step (a) follows from \eqref{condition F norm} and \eqref{condition O norm}, step (b) is due to Lemma \ref{lemma 1}, and the last inequality uses the fact that   if $\calP_{T_k^\perp}(\mM_k) = \bzero$ for all $1\leq k\leq K$, then $\mM_k = \bzero.$

\end{proof}

\subsection{Construction of Dual Certificate}
Following the well developed route in \cite{gross2011recovering}, the dual certificates $\{\mY_k\}_{k=1}^K$ are constructed as follows: for any $k\in[K]$, we firstly divide the linear measurements $\calA_k$ into $t_0$ partitions, denoted $\{\Omega_t\}_{t=1}^{t_0}$, and let $m=\frac{n}{t_0}$. Define 
\begin{align*}
	&\calA_{k,t}(\mX)=\left\{ \la \vb_{k, j}\ve_{j}^\tran,\mX\ra \right\}_{j\in\Omega_t} \in \C^{\lab\Omega_t\rab}.\numberthis\label{eq:partA}
\end{align*}
Let $\mY_k^0 = \bzero $ and $\blambda^0 := \sum_{k=1}^{K} \calA_{k,1}\calG^\ast(\mU_k \mV_k^\tranH)$. Then the golfing scheme for constructing $\{\mY_k\}$ are expressed as
\begin{align*}
\blambda^{t-1}&= \sum_{k=1}^{K} \calA_{k,t}\calG^\ast\calP_{T_k}(\mU_k \mV_k^\tranH - \mY_{k,t-1})\text{ and }\\
\mY_{k,t}&=\mY_{k, t-1} +\frac{n}{m} \calG\calA^\ast_{k,t} (\blambda^{t-1})\\
        &\quad + \left(\calI - \calG\calG^\ast\right)\calP_{T_k} \left( \mU_k \mV_k^\tranH - \mY_{k,t-1}\right), \\
        &\quad \text{ for }t=1,\cdots, t_0,\\
\mY_k &=\mY_{k,t_0}.        
\end{align*} 

\subsection{Validating the Dual Certificate and Completing the Proof}
In this section, we show that the constructed dual certificate satisfies the conditions \eqref{condition F norm} and \eqref{condition O norm}.  The derivation relies on a series of lemmas. Due to space constraints, a detailed proof of these lemmas will be provided in \cite{wang2023blind}.
\begin{lemma}
Let $t_0\in\{1,\cdots,n\}$ and set $m = \frac{n}{t_0}$. If $m\gtrsim K\sqrt{n}\sqrt{s\mu_0\mu_1r}\log(sn)$, then there exists a partition $\{\Omega_t\}_{t = 1}^{t_0}$ such that the following properties hold: 
\begin{align}
\label{eq: incoherece A z}
    &\opnorm{\calG\calA_{i,t}^\ast \calA_{j,t}\calG^\ast(\mM) } \lesssim \alpha,\\
    \label{eq: lemmawang1}
    &\left\| \calP_{T_i}\calG\calA_{i,t}^\ast \calA_{j,t}\calG^\ast\calP_{T_j} (\mM)\right\|_{\calG,\mathsf{F}}\lesssim \sqrt{\frac{\mu_1 r\log(sn)}{n}}\cdot\alpha,\\
    \label{eq: lemmawang2}
    &\left\| \calP_{T_i}\calG\calA_{i,t}^\ast \calA_{j,t}\calG^\ast\calP_{T_j} (\mM)\right\|_{\calG,\infty}\lesssim \frac{\mu_1 r}{n}\cdot\alpha,
\end{align} 
where $\alpha :=  \sqrt{s\mu_0 \log (sn)} \left\| \mM\right\|_{\calG, F} + s\mu_0\log(sn) \left\|\mM\right\|_{\calG,\infty}$, $\mM\in\C^{sn_1\times n_2}$ is fixed, and the definitions of  $\|\cdot\|_{\calG,\mathsf{F}}$ and $\|\cdot\|_{\calG,\infty}$ can be found in~\cite{chen2022vectorized}.
\end{lemma}

\begin{lemma}\label{lem: wang2}
Assume $n\gtrsim K\mu_0\mu_1sr\log(sn)$ and $m\gtrsim \frac{n}{K}$, for any $i\neq j$ and $1\leq t\leq t_0$,  the event
\begin{align*}
	\opnorm{\calP_{T_i}\calG\calA_{i,t}^\ast \calA_{j,t}\calG^\ast\calP_{T_j}} \leq \frac{1}{8K},
\end{align*}
occurs with probability at least $1-(sn)^{-c_1}$ for a universal constant $c_1>0$.
\end{lemma}
\begin{corollary}\label{cor: cor1}
    Assume $n\gtrsim K\mu_0\mu_1sr\log(sn)$ and $m\gtrsim \frac{n}{K}$, for any $i\neq j$,  the event
    \begin{align*}
	\opnorm{\calP_{T_i}\calG\calA_{i}^\ast \calA_{j}\calG^\ast\calP_{T_j}} \leq \frac{1}{8K},
\end{align*}
occurs with probability at least $1-(sn)^{-c_1}$ for a universal constant $c_1>0$.
\end{corollary}

Equipped with these lemmas, we turn to validate the conditions in Theorem~\ref{thm: unique theorem}. Note that the inequality~\eqref{eq: concentrate} is proved by Corollary 3.10 in \cite{chen2022vectorized}, and the inequality~\eqref{eq: incoh} follows from Corollary~\ref{cor: cor1}. It is not hard to see that \eqref{condition range} holds by the construction of $\{\mY_k\}_{k = 1}^K$. Hence, it remains to verify inequalities \eqref{condition F norm} and \eqref{condition O norm}.
\subsubsection{Validating~\eqref{condition F norm}} let $\mE_{k,t} := \mU_k \mV_k^\tranH - \calP_{T_k}(\mY_{k,t})$. A simple calculation yields that
 \begin{align*}
 \mY_{k,t} &= \mY_{k,t-1}  +\frac{n}{m} \calG\calA^\ast_{k,t} (\blambda^{t-1})+  \left(\calI - \calG\calG^\ast\right)(\mE_{k,t}),\\
     \mE_{k,t} 
	&=\left( \calP_{T_k}\calG\calG^\ast\calP_{T_k} -\frac{n}{m}  \calP_{T_k} \calG\calA^\ast_{k,t}  \calA_{k,t}\calG^\ast\calP_{T_k}\right)(\mE_{k,t-1})\\
 &\quad -\frac{n}{m}   \sum_{j\neq k}\calP_{T_k} \calG\calA^\ast_{k,t}  \calA_{j,t}\calG^\ast\calP_{T_j}(\mE_{j,t-1}).
 \end{align*}
Consequently, one can obtain
\begin{align*}
   & \fronorm{\mU_k\mV_k^\ast-\calP_{T_k}(\mY_k)} =\fronorm{\mE_{k,t_0}}\leq \max_{k\in [K]}\fronorm{\mE_{k,t_0}}\\
    &\leq \max_{k\in [K]} \fronorm{ \calP_{T_k}\lb\calG\calG^\ast -\frac{n}{m}   \calG\calA^\ast_{k,t_0}  \calA_{k,t_0}\calG^\ast\rb\calP_{T_k}(\mE_{k,t_0-1}) }\\
    &\quad + \fronorm{\frac{n}{m}   \sum_{j\neq k}\calP_{T_k} \calG\calA^\ast_{k,t_0}  \calA_{j,t_0}\calG^\ast\calP_{T_j}(\mE_{j,t_0-1}) }\\
    &\stackrel{(a)}{\leq } \max_{k\in [K]}\frac{1}{4}\fronorm{\mE_{k,t_0-1}}+\frac{1}{8K}\sum_{j\neq k}\fronorm{\mE_{j,t_0-1}}\\
    &\leq \max_{k\in [K]} \frac{1}{2}\fronorm{\mE_{k,t_0-1}} \leq \frac{1}{2^{t_0}}\max_{k\in [K]} \fronorm{\mE_{k,0}}\\
    & = \frac{1}{2^{t_0}}\max_{k\in [K]} \fronorm{\mU_k\mV_k^\tranH} \leq  \frac{r}{2^{t_0}}\leq \frac{1}{48Ks\mu_0},
\end{align*}
where inequality $(a)$ is due to Lemma 3.9 in \cite{chen2022vectorized}, Lemma~\ref{lem: wang2}, and  the last inequality holds when $t_0 = \lceil \log_2(48Krs\mu_0)\rceil$.

\subsubsection{Validating~\eqref{condition O norm}} a simple calculation yields that
\begin{align*}
    \mY_k = \sum_{t =1}^{t_0}\lsb \frac{n}{m}\calG\calA^{\ast}_{k,t}\sum_{j = 1}^K\calA_{j,t}\calG^\ast(\mE_{j,t-1})+(\calI-\calG\calG^\ast)\mE_{k,t-1}\rsb
\end{align*}
Applying triangular inequality gives 
\begin{align*}
    \opnorm{\calP_{T_k}^\perp (\mY_k) } &\leq \sum_{t =1}^{t_0} \opnorm{\lb \frac{n}{m}\calG\calA^{\ast}_{k,t}\calA_{k,t}\calG^\ast-\calG\calG^\ast\rb(\mE_{k,t-1})}\\
    &\quad +\sum_{t =1}^{t_0}\sum_{j\neq k}^K\opnorm{\frac{n}{m}\calG\calA^{\ast}_{k,t}\calA_{j,t}\calG^\ast(\mE_{j,t-1})}.
\end{align*}
For the sake of simplicity, we define the term $\alpha_{k,t}$  as
\begin{align*}
    \alpha_{k,t} := \sqrt{s\mu_0 \log (sn)} \left\| \mE_{k,t}\right\|_{\calG, \mathsf{F}} + s\mu_0\log(sn) \left\|\mE_{k,t}\right\|_{\calG,\infty}.
\end{align*}
For any $1\leq t\leq t_0$, Lemma 3.11 in \cite{chen2022vectorized} and inequality \eqref{eq: incoherece A z}  implies that
\begin{align*}
    \opnorm{\lb \frac{n}{m}\calG\calA^{\ast}_{k,t}\calA_{k,t}\calG^\ast-\calG\calG^\ast\rb(\mE_{k,t-1})} &\lesssim   \frac{n}{m}\alpha_{k,t-1},\\
    \opnorm{\frac{n}{m}\calG\calA^{\ast}_{k,t}\calA_{j,t}\calG^\ast(\mE_{j,t-1})} &\lesssim  \frac{n}{m}\alpha_{j,t-1}.
\end{align*}
Thus,  one has
\begin{align}\label{eq: t-1 bound}
    \opnorm{\calP_{T_k}^\perp (\mY_k) } &\lesssim \frac{n}{m}\sum_{t =1}^{t_0}\sum_{j= 1}^K  \alpha_{j,t-1}.
\end{align}
Recall that
\begin{align*}
    \mE_{j,t-1} &=\left( \calP_{T_j}\calG\calG^\ast\calP_{T_j} -\frac{n}{m}  \calP_{T_j} \calG\calA^\ast_{j,t}  \calA_{j,t}\calG^\ast\calP_{T_j}\right)(\mE_{j,t-2})\\
 &\quad -\frac{n}{m}   \sum_{k\neq j}\calP_{T_j} \calG\calA^\ast_{j,t}  \calA_{k,t}\calG^\ast\calP_{T_k}(\mE_{k,t-2}).
\end{align*}
Applying Lemma 3.12 in \cite{chen2022vectorized} and inequality \ref{eq: lemmawang1} yields that
\begin{align}\label{eq: t-2 bound gf}
    \left\| \mE_{j,t-1}\right\|_{\calG, \mathsf{F}} \lesssim  \sqrt{\frac{\mu_1r\log(sn)}{n}}\frac{n}{m}\sum_{k = 1}^K \alpha_{k,t-2}.
\end{align}
Utilizing Lemma 3.13 in \cite{chen2022vectorized} and inequality~\eqref{eq: lemmawang2}, the term $\left\| \mE_{j,t-1}\right\|_{\calG, \infty}$ is bounded by
\begin{align}\label{eq: t-2 bound gi}
    \left\|\mE_{j,t-1}\right\|_{\calG,\infty} \lesssim  \frac{\mu_1r}{n}\frac{n}{m}\sum_{k = 1}^K \alpha_{k,t-2}.
\end{align}
After substituting \eqref{eq: t-2 bound gf} and \eqref{eq: t-2 bound gi} into \eqref{eq: t-1 bound}, we have
\begin{align*}
    &\sum_{j = 1}^K\alpha_{j,t-1} \\
    &\lesssim \frac{nK\log(sn)}{m}\lb \sqrt{\frac{s\mu_0\mu_1r}{n}}+\frac{s\mu_0\mu_1r}{n}\rb \sum_{j = 1}^K\alpha_{j,t-2} \\
    &\stackrel{(a)}{\leq }  \frac{1}{2}\sum_{j = 1}^K\alpha_{j,t-2} \leq \lb\frac{1}{2}\rb^{t-1} \sum_{j = 1}^K\alpha_{j,0}.
\end{align*}
where  inequality $(a)$ is due to $m\gtrsim K\sqrt{n}\sqrt{s\mu_0\mu_1r}\log(sn)$.

Finally, noting that $\mE_{j,0} = \mU_j\mV_j^{\tranH}$, applying Lemma 3.14 in \cite{chen2022vectorized} yields that
\begin{align*}
     \opnorm{\calP_{T_k}^\perp (\mY_k) } &\lesssim \frac{n}{m}\sum_{t =1}^{t_0}\sum_{j= 1}^K  \alpha_{j,t-1}\\
     &\lesssim \frac{n}{m}\sum_{t =1}^{t_0}\lb\frac{1}{2}\rb^{t-1} \sum_{j = 1}^K\alpha_{j,0}\\
     &\lesssim  \frac{nK\log(sn)}{2m}\lb \sqrt{\frac{s\mu_0\mu_1r\log(sn)}{n}}+ \frac{s\mu_0\mu_1 r}{n}\rb \\
     &\leq \frac{1}{2},
\end{align*}
when $m\gtrsim \sqrt{n}K\sqrt{s\mu_0\mu_1r}\log(sn)$.

\section{Conclusion}
This paper studies the problem of simultaneous blind demixing and super-resolution. We introduce a convex approach named MVHL and rigorously establish its recovery performance. Our analysis demonstrates that MVHL can achieve exact recovery of the target matrices provided the sample complexity satisfying $n\gtrsim Ksr\log(sn)$. Furthermore, we illustrate the efficacy of MVHL through a series of numerical simulations.

\bibliographystyle{IEEEbib}

\bibliography{refBlind}

\end{document}